\documentclass[11pt]{article}
\usepackage[utf8]{inputenc}
\usepackage{amsmath}
\usepackage{mathtools}
\usepackage{amsfonts}
\usepackage{amssymb}
\usepackage{amsthm}
\usepackage[noabbrev]{cleveref}
\usepackage{comment}
\usepackage{algorithm}
\usepackage{algpseudocode}
\usepackage[nodisplayskipstretch]{setspace}
\usepackage{url}
\usepackage{cite}
\usepackage{enumitem}  
\usepackage[normalem]{ulem}

\usepackage{graphics} 
\usepackage{epsfig} 
\usepackage{times} 
\usepackage[margin=1in]{geometry}
\usepackage[numbers,sort,compress]{natbib}

\newtheorem{lemma}{Lemma}

\newtheorem{assumption}{Assumption}
\newtheorem{remark}{Remark}

\usepackage{xcolor}

\begin{document}

\title{Distributed and Localized Model Predictive Control. Part II: Theoretical Guarantees}
%
%

\author{Carmen~Amo~Alonso,
~ Jing Shuang (Lisa) Li,
~ Nikolai~Matni,
~and~James~Anderson
\thanks{C. Amo Alonso and J.S. Li are with the Computing and Mathematical Sciences Department at California Institute of Technology, Pasadena, CA, 91106, USA. N. Matni is with the Department of Electrical and Systems Engineering at the University of Pennsylvania, Philadelphia, PA, 19104, USA. J. Anderson is with the Department of Electrical Engineering and the Data Science Institute at Columbia University, New York, NY, 10027, USA. {\tt\small \{camoalon,jsli\}@caltech.edu}, {\tt\small nmatni@seas.upenn.edu}, {\tt\small james.anderson@columbia.edu}}%
}

\markboth{}
{Amo Alonso \MakeLowercase{\textit{et al.}}: Distributed and Localized Model Predictive Control: Theoretical Guarantees}

\maketitle




\begin{abstract}

Engineered cyberphysical systems are growing increasingly large and complex. These systems require scalable controllers that robustly satisfy state and input constraints in the presence of additive noise -- such controllers should also be accompanied by theoretical guarantees on feasibility and stability. In our companion paper, we introduced Distributed and Localized Model Predictive Control (DLMPC) for large-scale linear systems; DLMPC is a scalable \emph{closed-loop} MPC scheme in which subsystems need only exchange local information in order to synthesize and implement local controllers. In this paper, we provide recursive feasibility and asymptotic stability guarantees for DLMPC. We leverage the System Level Synthesis framework to express the maximal positive robust invariant set for the closed-loop system and its corresponding Lyapunov function, both in terms of the closed-loop system responses. We use the invariant set as the terminal set for DLMPC, and show that this guarantees feasibility with minimal conservatism. We use the Lyapunov function as the terminal cost, and show that this guarantees stability. We provide fully distributed and localized algorithms to compute the terminal set offline, and also provide necessary additions to the online DLMPC algorithm to accommodate coupled terminal constraint and cost. In all algorithms, only local information exchanges are necessary, and computational complexity is independent of the global system size -- we demonstrate this analytically and experimentally. This is the first distributed MPC approach that provides minimally conservative yet fully distributed guarantees for recursive feasibility and asymptotic stability, for both nominal and robust settings.

\end{abstract}



\section{Introduction}

Model Predictive Control (MPC) enjoys widespread success across diverse applications. Ensuring recursive feasibility and asymptoptic stability for MPC is a well-studied topic in the centralized setting \cite{borrelli_predictive_2017}, and sufficient conditions based on terminal sets and cost functions have been established \cite{mayne_robust_2005}. Porting these ideas to distributed systems is a challenging task, both theoretically and computationally. High computational demand, limited and local communication, and coupling among subsystems prevents the use of techniques from the centralized setting. Thus, efforts have been made to develop theoretical guarantees for distributed MPC.

\textbf{Prior work:} The majority of distributed MPC approaches rely on the use of distributed terminal costs and terminal sets to provide theoretical guarantees. In order to obtain structure in the terminal cost, standard methods rely on Lyapunov stability results, often combined with relaxation techniques to make them amenable to distributed settings \cite{langbort_distributed_2004,jokic_decentralized_2009,zecevic_control_2010}. For terminal sets, proposed distributed methods are often limited by the coupling among subsystems -- this often leads to small terminal sets  that result in too conservative solutions (see for example \cite{stewart_cooperative_2010,maestre_distributed_2011} and references therein). In order to overcome these issues, several approaches have recently been proposed to synthesize structured terminal sets with adaptive properties, i.e. local terminal sets defined as the sub-level set of a structured Lyapunov function, which change at each iteration in order to avoid unnecessary conservatism \cite{conte_distributed_2016,trodden_distributed_2017,darivianakis_distributed_2020,aboudonia_distributed_2020,muntwiler_distributed_2020,wang_robust_2021}. These approaches successfully design structured robust positive invariant sets and reduce conservatism; however, they require online updates of the terminal set at each MPC iteration, which increase the controller's overall computational complexity and communication overhead. Moreover, the imposed structure unavoidably leads to a possibly small approximation of the maximal control invariant set (the least conservative option for a terminal set). To move away from structured sets and costs, a data-driven approach was recently developed in \cite{sturz_distributed_2020}, where locally collected data are used to construct local control invariant sets and costs that provide guarantees. However, this approach is mainly limited to iterative control tasks, and conservatism of the terminal cost and set only reduces asymptotically as the system collects data. Online computation and refinements of the terminal set are also key to this approach. 

Given the state of the art, our goal is to design a distributed MPC approach with \emph{distributed} and \emph{minimally conservative} feasibility and stability guarantees. We seek a distributed MPC algorithm with (i) a maximal positive invariant terminal set, and (ii) a fully distributed and scalable \textit{offline} algorithm to compute this set. This will allow us to use the associated Lyapunov function of the terminal set as the terminal cost, and requires no explicit a priori structural assumptions. No method satisfying these requirement currently exists in the literature. 

\textbf{Contributions:} We provide theoretical guarantees for the Distributed Localized MPC (DLMPC) for linear time-invariant systems approach presented in our companion paper \cite{amoalonso_implementation_2021}. We show that the \emph{maximal} positive invariant set of the closed-loop system can be expressed in terms of the closed-loop system responses  as defined in the System Level Synthesis (SLS) framework \cite{wang_system_2019,anderson_system_2019}. We show that when the closed-loop system is localized, the set is naturally structured without requiring additional assumptions. We also show that this set can be used to provide recursive feasibility guarantees when used as the terminal set of the system, and stability guarantees when combined with its associated global Lyapunov function \cite{blanchini_set_1999}. We provide a fully distributed and localized offline algorithm for computation of the terminal set -- this algorithm requires only local information exchange between subsystems. We also provide necessary additions to the original DLMPC algorithm to accommodate coupled terminal constraint and cost. In particular, this can be done by using a nested Alternating Direction Method of Multipliers (ADMM)-based consensus algorithm. In the resulting implementation, each sub-controller first solves for its local portion of the terminal set, offline, then solves a local online MPC problem. Throughout all algorithms, only local information exchanges within some local neighborhood are necessary, and computational complexity is independent of the global system size. The presented approach applies to the nominal case as well as additive polytopic or locally norm-bounded disturbances. This approach is the first to compute without approximation the maximal positive invariant set and its associated global Lyapunov function in a fully distributed and localized manner. Through numerical experiments, we validate these results and further confirm the minimal conservatism introduced by this method.

\textbf{Paper structure:} In \S II we present the problem formulation and briefly summarize essential concepts from our companion paper \cite{amoalonso_implementation_2021}. In \S III, we formulate the maximal robust positive invariant set and its associated Lyapunov function in closed-loop coordinates, and use this to provide recursive feasibility and stability guarantees for DLMPC; we also discuss convergence guarantees for DLMPC. In \S IV we provide an offline algorithm to distributedly and locally compute the terminal set, and a modified online DLMPC algorithm that accommodates the terminal set and cost, including any local coupling. In \S V, we present a numerical study and we end in \S VI with conclusions and directions of future work.
 
\textbf{Notation:} Lower-case and upper-case Latin and Greek letters such as $x$ and $A$ denote vectors and matrices respectively, although lower-case letters might also be used for scalars or functions (the distinction will be apparent from the context). Bracketed indices denote time-step of the real system, i.e., the system input is $u(t)$ at time $t$, not to be confused with $x_t$ which denotes the predicted state $x$ at time $t$. Superscripted variables, e.g. $x^k$, correspond to the value of $x$ at the $k^{th}$ iteration of a given algorithm. Square bracket notation, i.e., $[x]_{i}$ denotes the components of $x$ corresponding to subsystem $i$. Calligraphic letters such as $\mathcal{S}$ denote sets, and lowercase script letters such as $\mathfrak{c}$ denote a subset of $\mathbb{Z}^{+}$, e.g. $\mathfrak{c}=\left\{1,...,n\right\}\subset\mathbb{Z}^{+}$.  Boldface lower and upper case letters such as $\mathbf{x}$ and $\mathbf{K}$ denote finite horizon signals and block lower triangular (causal) operators, respectively:
\begin{equation*} 
\mathbf{x}=\left[\begin{array}{c} x_{0}\\x_{1}\\\vdots\\x_{T}\end{array}\right], ~
K =   { {\scriptscriptstyle{\left[\begin{array}{cccc}K_{0}[0] & & & \\ K_{1}[1] & K_{1}[0] & & \\ \vdots & \ddots & \ddots & \\ K_{T}[T] & \dots & K_{T}[1] & K_{T}[0] \end{array}\right]}}},
\end{equation*}
where each $x_i$ is an $n$-dimensional vector, and each $K_{i}[j]$ is a matrix of compatible dimension representing the value of $K$ at the $j^\text{th}$ time-step computed at time $i$. $\mathbf{K}(\mathfrak{r},\mathfrak{c})$ denotes the submatrix of $\mathbf{K}$ composed of the rows and columns specified by $\mathfrak{r}$ and $\mathfrak{c}$ respectively. We denote the block columns of  $\mathbf K$ by $\mathbf K\{1\}$,...,$\mathbf K\{T\}$, i.e. $\mathbf K\left\{1\right\}:=[K_{0}[0]^{\mathsf{T}}\ \dots\ K_{T}[T]^{\mathsf{T}}]^{\mathsf{T}}$, and we use $:$ to indicate the range of columns, i.e. $\mathbf K\left\{2:T\right\}$ contains the block columns from the second to the last. For compactness, we also define $Z_{AB}:=\begin{bmatrix} I-Z\hat A & -Z\hat B \end{bmatrix}$ where $\hat A:=\mathrm{blkdiag}(A,...,A)$ and $\hat B:=\mathrm{blkdiag}(B,...,B,0)$ for the dynamics matrices $A$ and $B$, and $Z$ is the block-downshift matrix.

\section{Problem Formulation}

We begin with a brief summary of the DLMPC formulation, approach, and algorithm (for details, refer to our companion paper \cite{amoalonso_implementation_2021}), then formally introduce the problem of providing theoretical guarantees -- recursive feasibility and asymptotic stability -- in the DLMPC scheme. We then introduce the open question that we resolve in this paper. 

\textbf{Setup:}
Consider a discrete-time linear time invariant (LTI) system
\begin{equation} \label{eqn:system}
x(t+1) = Ax(t)+Bu(t)+w(t),
\end{equation}
where $x(t)\in\mathbb{R}^{n}$ is the state, $u(t)\in\mathbb{R}^{p}$ is the control input, and $w(t)\in\mathcal W \subset \mathbb{R}^{n}$ is an exogenous disturbance at time $t$. System \eqref{eqn:system} can be interpreted as $N$ interconnected subsystems. Each subsystem is equipped with a sub-controller. We model the interconnection topology as an unweighted directed graph $\mathcal{G}_{(A,B)}(\mathcal E,\mathcal V)$, where each subsystem $i$ is identified with a vertex $v_{i}\in \mathcal V$ and an edge $e_{ij}\in \mathcal E$ exists whenever $[A]_{ij}\neq 0$ or $[B]_{ij}\neq 0$.  

We impose that information exchange between sub-controllers -- as defined by the graph $\mathcal{G}_{(A,B)}$ -- is confined to a subset of neighboring sub-controllers. We use the $d$-local communication constraints \cite{wang_separable_2018} to formalize this idea. Each subsystem $i$:
\begin{itemize}
\item Receives information from its \textit{d-incoming set} $\textbf{in}_{i}(d) := \left\{v_{j}\ |\ \textbf{dist}(v_{j} \rightarrow v_{i} ) \leq d\in\mathbb{N} \right\}$, and
\item Sends information to its \textit{d-outgoing set} $\textbf{out}_{i}(d) := \left\{v_{j}\ |\  \textbf{dist}(v_{i} \rightarrow v_{j} ) \leq d\in\mathbb{N} \right\}$.
\end{itemize}

We use a model predictive controller to determine the control input; at time step $\tau$, the controller solves:
\begin{align} \label{eqn:MPC}
& \underset{{x}_{t},u_{t}, \gamma_t}{\text{min}} &  &\sum_{t=0}^{T-1}f_{t}(x_{t},u_{t})+f_{T}(x_{T})\\
& \ \text{s.t.} &  &\begin{aligned} \nonumber
&x_{0} = x(\tau),\ x_{t+1} = Ax_{t}+Bu_{t}+w_t, \\
&x_{T}\in\mathcal{X}_{T},\, x_{t}\in\mathcal{X}_{t},\, u_{t}\in\mathcal{U}_{t}\ \forall w_t\in\mathcal{W}_t, \\
&u_{t} = \gamma_t(x_{0:t},u_{0:t-1}),\ t=0,...,T-1.
\end{aligned}
\end{align}
To provide tractability, $f_t(\cdot,\cdot)$ and $f_T(\cdot)$ are assumed to be closed, proper, and convex, and $\gamma_t(\cdot)$ is a measurable function of its arguments. The sets $\mathcal{X}_t$ and $\mathcal{U}_t$ are assumed to be closed and convex sets containing the origin for all $t$. For simplicity, we will consider constant state and input constraint sets $\mathcal{X}$ and $\mathcal{U}$.

Note that \eqref{eqn:MPC} does not include the $d$-local communication constraints. It is not possible to introduce such constraints in a convex manner in the classical MPC formulation \eqref{eqn:MPC}. However, the DLMPC formulation allows to incorporate locality constraints into the MPC formulation in a straightforward manner with the only requirement that the MPC formulation be compatible with the locality constraints.(see our companion paper \cite{amoalonso_implementation_2021} for details). Hence, we assume that if two subsystems are coupled through either the constraints or cost, then the two subsystems must be in the $d$-incoming and $d$-outgoing set from one another:
\begin{assumption}{\label{assump:locality}}
Given an MPC problem \eqref{eqn:MPC} over system \eqref{eqn:system}, the objective function $f$ is such that $f(x,u)=\sum f^i([x]_{\textbf{in}_{i}(d)},[u]_{\textbf{in}_{i}(d)})$ for $f^i$ local functions; and the constraint sets are such that $x\in\mathcal{X}=\mathcal{X}^1\cap ... \cap \mathcal{X}^N$, where $x \in \mathcal{X}$ if and only if $[x]_{\textbf{in}_{i}(d)}\in\mathcal{X}^i$ for all $i$, and idem for $\mathcal{U}$ for the local sets $\mathcal{X}^i$ and $\mathcal{U}^i$.
\end{assumption}

\textbf{Approach:}
In \cite{amoalonso_implementation_2021}, we use the SLS framework to reformulate the MPC problem \eqref{eqn:MPC} into the DLMPC problem. This allows for distributed and localized synthesis and implementation, i.e., each subsystem requires only local information to synthesize its local sub-controller and determine the local control action. This is made possible by imposing appropriate $d$-local structural constraints $\mathcal{L}_d$ on the \emph{closed-loop system responses} of the system $\mathbf \Phi_x$ and $\mathbf \Phi_u$, which become the decision variables of the MPC problem. 

The DLMPC subroutine over time horizon $T$ at time $\tau$ is as follows:
\begin{align}
& \underset{\mathbf{\Phi}}{\text{min}} & f(\mathbf{\Phi}\{1\}x_{0}) + f_T(\mathbf{\Phi}\{1\}x_{0}) \label{eqn:dlmpc} \hspace{1.2cm}\\
& \text{~s.t.} &\begin{aligned} & Z_{AB}\mathbf{\Phi}=I, ~ x_0 = x(\tau), ~ \mathbf{\Phi}\in\mathcal L_d,\\
					    & \mathbf{\Phi}\mathbf w\in \mathcal{P}, ~ \mathbf\Phi_{x,T}\mathbf w\in\mathcal X_T, ~\forall\mathbf w \in\mathcal W,
       		     \end{aligned} \nonumber
\end{align}
where $\mathbf \Phi := \begin{bmatrix}\mathbf\Phi_x^\intercal & \mathbf\Phi_u^\intercal \end{bmatrix}^\intercal$, $f_t(\cdot,\cdot)$ and $f_T(\cdot)$ are closed, proper, and convex cost functions, and $\mathcal P$ is defined so that $\mathbf{\Phi}\mathbf{w}\in\mathcal{P}$ if and only if $\mathbf{x}\in\mathcal{X},\text{ and } \mathbf u\in\mathcal{U}$. By convention, we define the disturbance to contain the initial condition, i.e., $\mathbf{w} = [x_{0}^\intercal\ w_{0}^\intercal\ \dots \ w_{T-1}^\intercal]=:[x_{0}^\intercal\ {\boldsymbol{\delta}^\intercal}]$ and $\mathcal W$ is defined over $\boldsymbol{\delta}$ so that it does not restrict $x_0$. In the nominal case, i.e., $w_t=0\ \forall t$, $\mathcal{P}$ and $\mathcal{X}_T$ are closed and convex sets containing the origin. When noise is present, we restrict ourselves to polytopic sets only: $\mathcal P:= \{[\mathbf{x}^\intercal \ \mathbf{u}^\intercal]^\intercal:\ H[\mathbf{x}^\intercal \ \mathbf{u}^\intercal]^\intercal\leq h\}$, and consider different options for set $\mathcal{W}$: 
\begin{itemize}
\item Polytopic set: $\boldsymbol \delta \in\{\boldsymbol \delta :\ G\boldsymbol \delta \leq g\}$.
\item Locally norm-bounded: $[\boldsymbol \delta]_i\in\{ \boldsymbol \delta \, : \, \| \boldsymbol \delta \|_p \leq \sigma\}$  $\forall i=1,...,N$ and $p\geq 1$.
\end{itemize}

Although \eqref{eqn:dlmpc} solves a distributed control problem, the optimization problem \eqref{eqn:dlmpc} is itself centralized. To solve \eqref{eqn:dlmpc} in a distributed and localized manner, we perform a variable duplication in order to apply the ADMM distributed optimization technique \cite{boyd_distributed_2010}:
\begin{align}
& \underset{\mathbf{\tilde\Phi},\mathbf{\tilde\Psi}}{\text{min}} & f(M_1\mathbf{\tilde \Phi}\{1\}x_{0})   \label{eqn:dlmpc_admm} & ~\\
&  \text{s.t.} & Z_{AB}M_2\mathbf{\tilde \Psi}=I, & ~ x_0 = x(\tau), ~ \mathbf{\tilde\Phi},\mathbf{\tilde \Psi}\in\mathcal{L}_d, \nonumber \\
& ~ &\mathbf{\tilde\Phi}x_0\in\mathcal{\tilde P},\ \mathbf{\tilde \Phi}&=\tilde H\mathbf{\tilde\Psi} \nonumber
\end{align}
where variables $\mathbf{\tilde \Phi}$ and $\mathbf{\tilde \Psi}$ are extensions of the system response $\mathbf\Phi$ to account for the robust case and $M_1,\ M_2,\ \tilde H$ are auxiliary matrices.\footnote{Their actual definition depends on whether the system experiences no disturbances, locally bounded disturbances, or polytopic disturbances; they are defined in \S IV of \cite{amoalonso_implementation_2021}.} We apply ADMM to solve \eqref{eqn:dlmpc_admm} in a distributed and localized manner, as shown in Algorithm \ref{alg:dlmpc}: 
\setlength{\textfloatsep}{0pt}
\begin{algorithm}[ht]
\caption{Subsystem $i$ DLMPC implementation}\label{alg:dlmpc}
\begin{algorithmic}[1]
\State Measure local state $[x(\tau)]_{i}$ and exchange with neighbors in $\textbf{out}_{i}(d)$. Set $k\leftarrow0$. 
\State Share the measurement with neighbors in $\textbf{out}_{i}(d)$.
\State Solve for local rows of $\mathbf{\tilde\Phi}^{k+1}$ via (11a) in \cite{amoalonso_implementation_2021}.
\State Exchange rows of $\mathbf \Phi$ with $d$-local neighbors. 
\State Solve for local columns of $\mathbf{\tilde\Psi}^{k+1}$ via (11b) in \cite{amoalonso_implementation_2021}.
\State Exchange rows of $\mathbf \Phi$ with $d$-local neighbors. 
\State Perform the multiplier update step via (11c) in \cite{amoalonso_implementation_2021}.
\State \textbf{if} ADMM has converged\textbf{:}
\Statex $\;\;$ Apply $[u_0]_i = [\Phi_{u,0}[0]]_{i}[x_0]_{i}$. Return to step 1.
\Statex \textbf{else:} 
\Statex $\;\;$ Set $k\leftarrow k+1$. Return to step 3.
\end{algorithmic}
\end{algorithm}

\textbf{Problem statement:}
The DLMPC approach \eqref{eqn:dlmpc_admm} introduced in \cite{amo_alonso_distributed_2020,amo_alonso_robust_2021} lacks feasibility, stability, and convergence guarantees -- our goal is to provide these. Also, \eqref{eqn:dlmpc_admm} does not explicitly consider a terminal set $\mathcal X_T$ and terminal cost $f_T$ -- we want to leverage these to provide the aforementioned feasibility and stability guarantees. Additionally, Algorithm \ref{alg:dlmpc} was developed under the assumption that subsystems are not coupled. We want to augment the algorithm to accommodate coupling as per Assumption \ref{assump:locality}. In the remainder of this paper, we address all of these problems; we provide theoretical guarantees by selecting an appropriate terminal set and cost, and present distributed and localized algorithms that perform the necessary computations and accommodate coupling.

\section{Feasibility and Stability Guarantees}\label{sec:guarantees}

Theoretical guarantees for the DLMPC problem \eqref{eqn:dlmpc} are now derived. First, we describe a maximal positive invariant set using an SLS-style parametrization; recursive feasibility for DLMPC is guaranteed by using this set as the terminal set. We also use this set to construct a terminal cost to guarantee asymptotic stability for the nominal setting and input-to-state stability (ISS) for the robust setting. Convergence results from the ADMM literature are used to establish convergence guarantees.

\subsection{Feasibility guarantees}

Recursive feasibility guarantees for the DLMPC problem \eqref{eqn:dlmpc} are given by the following lemma:

\begin{lemma}\label{lemm:terminal_set}
Let the terminal set $\mathcal X_T$ for the DLMPC problem \eqref{eqn:dlmpc} be of the form 
\begin{equation}\label{eqn:terminal_set}
\mathcal{X}_T:=\{x_0\in\mathbb R^n:\ \mathbf \Phi \begin{bmatrix} x_0^\intercal & \boldsymbol\delta^\intercal \end{bmatrix}^\intercal \in \mathcal P\ \forall \boldsymbol\delta\in\mathcal W\},
\end{equation} 
for some $\mathbf\Phi$ satisfying $Z_{AB}\mathbf{\Phi} = I$. Then recursive feasibility is guaranteed for the DLMPC problem \eqref{eqn:dlmpc}. Moreover, $\mathcal X_T$ is the maximal robust positive invariant set for the closed-loop described by $\mathbf \Phi$.
\end{lemma}

\begin{proof}

First, we show that $\mathcal X_T$ is the maximal robust positive invariant set. By Algorithm 10.4 in \cite{borrelli_predictive_2017}, the set
$$\mathcal S :=\bigcap_{k=0}^\infty \mathcal S_k,\quad \text{with}\quad \mathcal S_0 = \mathcal X,$$ \vspace{-5mm}
$$\mathcal S_k = \{x\in\mathbb R^n:\ Ax+Bu+w\in\mathcal S_{k-1}\ \forall w\in\mathcal W\}$$
is the maximal robust positive invariant set for the closed-loop system \eqref{eqn:system} with some fixed input $u\in\mathcal U$. We show by induction that $\mathcal S_k$ can also be written as
\begin{equation}\label{eqn:induction}
\mathcal S_k=\{x_0\in\mathbb R^n:\ x_{k}\in\mathcal X\ \forall w_{0:k-1}\in\mathcal W\}
\end{equation}
for some sequence $u_{0:k-1}\in\mathcal U$. The base case at $k=1$ is trivially true, since $\mathcal S_0 = \mathcal X$. For the inductive step, assume that equation \eqref{eqn:induction} holds at $k$. Then, at $k=1$,
\begin{align*}
\mathcal S_{k+1} &= \big\{x\in\mathbb R^n:\ Ax+Bu+w\in\mathcal S_k\ \forall w\in\mathcal W\big\} \\
& = \big\{x_0\in \mathbb R^n:\ x:=Ax_0+Bu+w\ s.t. \\
&\hspace{4mm}A^{k+1}x+\sum_{j=0}^{k}A^j(Bu_j+w)\in\mathcal X\ \forall w\in\mathcal W \big\} \\
&= \big\{x_0\in \mathbb R^n:\ x_{k+1}\in\mathcal X,\ \forall w_{0:k}\in\mathcal W \big\}
\end{align*}
for some sequence of inputs $u_{0:k}\in\mathcal U$.
This implies that $\mathcal S$ can be written as 
$$\mathcal S=\{x_0\in\mathbb R^n:\ \mathbf{x}\in\mathcal X\ \forall w\in\mathcal W\} \text{ for some } \mathbf{u}\in\mathcal U.$$
From the definition of $\mathbf \Phi$, this implies directly that
$$ \mathcal{S}=\{x_0\in\mathbb R^n:\ \mathbf \Phi \begin{bmatrix} x_0^\intercal & \boldsymbol\delta^\intercal \end{bmatrix}\in \mathcal P\ \forall \boldsymbol\delta\in\mathcal W \}:=\mathcal X_T, $$
for some $\mathbf\Phi$ satisfying $Z_{AB}\mathbf\Phi=I$. This constraint automatically enforces that the resulting closed-loop is feasible, i.e. that $\mathbf u = \mathbf{\Phi}_u \begin{bmatrix} x_0^\intercal & \boldsymbol\delta^\intercal \end{bmatrix}$ exists; furthermore, all inputs $\mathbf u$ can be captured by this closed-loop parametrization since it parametrizes a linear time-varying controller over a finite time horizon \cite{anderson_system_2019}. Hence, $\mathcal X_T$ is the maximal robust positive invariant set for the closed-loop system as defined by $\mathbf\Phi$.

Next, we show that imposing $\mathcal X_T$ as the terminal set of the DLMPC problem \eqref{eqn:dlmpc} guarantees recursive feasibility. Notice that by definition, $\mathcal X_T$ is not only a robust positive invariant set but also a robust \textit{control} invariant set. We can directly apply the proofs from Theorem 12.1 in \cite{borrelli_predictive_2017} to the robust setting -- recursive feasibility is guaranteed if the terminal set is control invariant, which $\mathcal X_T$ is. 
\end{proof}

\begin{remark}
Recursive feasibility is guaranteed by a \emph{control} invariant $\mathcal X_T$. In the ideal case, we want $\mathcal X_T$ to be the maximal robust control invariant set, in order to minimize conservatism introduced by $\mathcal X_T$ in \eqref{eqn:dlmpc}. However, this set generally lacks sparsity, violates Assumption \ref{assump:locality}, and is not amenable for inclusion in our distributed and \textit{localized} algorithm. For this reason, we use the maximal robust \emph{positive} invariant set in Lemma \ref{lemm:terminal_set}. Here, the choice of $\mathbf\Phi$ is critical in determining how much conservatism $\mathcal X_T$ will introduce. As suggested in \S 12 of \cite{borrelli_predictive_2017}, we choose $\mathbf\Phi$ corresponding to the unconstrained closed-loop system. In the interests of distributed synthesis and implementation, we additionally enforce $\mathbf\Phi$ to have localized structure. We discuss how this $\mathbf\Phi$ is computed in \S \ref{sec:implementation}, and demonstrate that the resulting terminal set $\mathcal X_T$ introduces no conservatism in \S \ref{sec:simulation}.
\end{remark}

By Assumption \ref{assump:locality}, constraint sets are localized, i.e. $\mathcal{X}=\mathcal{X}^1\cap ... \cap \mathcal{X}^N,$ where $x \in \mathcal{X}$ if and only if $[x]_{\textbf{in}_{i}(d)}\in\mathcal{X}^i$ for all $i$ (and idem for $\mathcal U$). Moreover, we can use the constraint $\mathcal L_d$ to enforce that the system response $\mathbf\Phi$ is localized. This implies that the set $\mathcal X_T$ is also localized: 
$$\mathcal{X}=\mathcal{X}_T^1\cap ... \cap \mathcal{X}_T^N,$$ 
where 
$$\hspace{-3mm}\mathcal{X}_T^i = \{[x_0]_{\textbf{in}_{i}(d)}\in\mathbb R^{[n]_i}: [\mathbf \Phi]_{\textbf{in}_{i}(d)} \begin{bmatrix} x_0^\intercal & \boldsymbol\delta^\intercal \end{bmatrix}_{\textbf{in}_{i}(2d)}\in \mathcal P^i\ \ \forall [\boldsymbol\delta]_{\textbf{in}_{i}(2d)} \in\mathcal W^{\textbf{in}_{i}(d)}\} $$
and $x \in \mathcal{X}_T$ if and only if $[x]_{\textbf{in}_{i}(d)}\in\mathcal{X}_T^i$ for all $i$. We will show in \S \ref{sec:implementation} that this allows for a localized and distributed computation of the terminal set $\mathcal{X}_T$.

\subsection{Stability guarantees}

Stability guarantees for the DLMPC problem \eqref{eqn:dlmpc} are given by the following lemma:

\begin{lemma}\label{lemm:stability}
	Consider system \eqref{eqn:system} subject to the MPC law \eqref{eqn:dlmpc}, where:
	\begin{enumerate}
		\item The cost $f$ is continuous and positive definite.
		\item The set $\mathcal P$ contains the origin and is closed.		
		\item $\mathcal{X}_{T}$ is defined by \eqref{eqn:terminal_set}.
		\item $f_{T}(x) = \text{inf}\ \{\eta\geq 0: \ x\in\eta\mathcal X_T\}.$
	\end{enumerate}
	Then, in the nominal setting, the origin is asymptotically stable with domain of attraction $\mathcal X$, and in the robust setting, $\mathcal X_T$ is input-to-state stable with domain of attraction $\mathcal X$.	
\end{lemma}

\begin{proof}
It suffices to show that these conditions immediately imply satisfaction of the necessary assumptions in \cite{mayne_constrained_2000}, together with the additional sufficient conditions of Theorem 4.2 in \cite{lofberg_minimax_2003}. These results state that if 
\begin{enumerate}[label=(\roman*)]
	\item $f$ and $f_T$ are continuous and positive definite, 
	\item $\mathcal X,\ \mathcal U,\ \mathcal X_T$ contain the origin and are closed, 
	\item $\mathcal X_T$ is control invariant, and 
	\item $\underset{u\in\mathcal U}{\text{min}}\ f(x,u)+f_T(Ax+Bu)-f_T(x)\leq0\ \forall x\in\mathcal X_T,$ 
\end{enumerate}
then the desired stability guarantees hold.
	
Condition \textit{1)} implies satisfaction of the part of (i) that concerns $f$. 
	
Condition \textit{2)} implies satisfaction of (ii). If $\mathcal P$ contains the origin and is closed, by definition this implies that $\mathcal X$ and $\mathcal U$ also contain the origin and are closed. Also, since $\mathcal X_T$ is defined in terms of $\mathcal P$ in \eqref{eqn:terminal_set}, $\mathcal X_T$ also contains the origin and is closed.
	
Condition \textit{3)} implies satisfaction of (iii) by virtue of Lemma \ref{lemm:terminal_set}.
	
Condition \textit{4)} implies that $f_T$ is a Lyapunov function on $\{x\in\mathbb R^{n}:\ 1\leq f_T(x)\}\supseteq\mathcal X_T$ since it is the Minkowski functional of the terminal set $\mathcal X_T$ (see Theorem 3.3. in \cite{blanchini_set_1999}). Therefore, the condition stated in (iv) is automatically satisfied for all $x\in\mathcal X_T$. Moreover, $f_T$ is necessarily positive definite, so the part of (i) that concerns $f_T$ is satisfied as well.
	
Therefore, by virtue of the results in \cite{mayne_constrained_2000} and \cite{lofberg_minimax_2003}, we guarantee asymptotic stability of the origin in the nominal setting and ISS of $\mathcal X_T$ in the robust setting, both with domain of attraction $\mathcal X$.	

\end{proof}

For any cost $f$ and constraint $\mathcal P$ that satisfy conditions \textit{1)} and \textit{2)} of Lemma \ref{lemm:stability}, we can choose an appropriate terminal set $\mathcal X_T$ and terminal cost $f_T$ as per \textit{3)} and \textit{4)} to satisfy the lemma. This guarantees stability for the DLMPC problem \eqref{eqn:dlmpc}. However -- as stated, $f_T$ does not satisfy Assumption \ref{assump:locality}, and therefore it is not localized. In particular, $f_T$ can be written as: 
\begin{equation}\label{eqn:terminal_cost}
f_T (x) = \underset{\eta}{\text{inf}}\ \{\eta\geq 0: \ [x]_{\textbf{in}_i(d)}\in\eta\mathcal X^i_T\ \forall i\},
\end{equation}
which cannot be written as a sum of local functions. Nonetheless, this scalar objective function admits a distributed and localized implementation -- we can add it in the DLMPC algorithm using the ADMM-based consensus technique described in \S \ref{sec:implementation}.

\subsection{Convergence guarantees}

Algorithm \ref{alg:dlmpc} relies on ADMM. We can guarantee convergence of the overall algorithm by leveraging the ADMM convergence result from \cite{boyd_distributed_2010}.

\begin{lemma}\label{lemm:convergence}

In Algorithm \ref{alg:dlmpc}, the residue, objective function, and dual variable converge as $k\rightarrow \infty$ i.e.
 $$\tilde{H}\mathbf{\tilde\Phi}^k-\mathbf{\tilde\Psi}^k \rightarrow 0, ~ f(\mathbf{\tilde\Phi}^k x_0)\rightarrow f(\mathbf{\tilde\Phi}^* x_0), ~ \mathbf{\tilde\Psi}^k\rightarrow\mathbf{\tilde\Psi}^*,$$
 where $\star$ indicates optimal value.

\end{lemma}

\begin{proof}
Algorithm \ref{alg:dlmpc} is the result of applying ADMM to the DLMPC problem \eqref{eqn:dlmpc}, then exploiting the separability and structure of resulting sub-problems to achieve distributed and localized implementation, as presented in \cite{amoalonso_implementation_2021}. Thus, to prove convergence, we only need to show that the underlying ADMM algorithm converges. By the ADMM convergence result in \cite{boyd_distributed_2010}, the desired convergence of the residue, objective function, and dual variable are guaranteed if
\textit{
\begin{enumerate}
	\item The extended-real-value functional for the algorithm is closed, proper, and convex, and
	\item The unaugmented Lagrangian for the algorithm has a saddle point.
\end{enumerate} 
}

We first show \textit{1)}. The extended-real-value functional $h(\mathbf{\tilde\Phi})$ is defined for this algorithm as
\begin{equation*}
    h(\mathbf{\tilde\Phi})=\begin{cases}
    f(M_1\mathbf{\tilde\Phi}\{1\}x_0) &\text{if $Z_{AB}M_2\tilde{H}^\dagger\mathbf{\tilde\Phi}=I,\ $} \\
    ~ &\mathbf{\tilde\Phi}\in\mathcal L_d, \mathbf{\tilde\Phi}x_0\in\mathcal{\tilde P},\\
    \infty &\text{otherwise}. \end{cases}
\end{equation*}

\noindent where $\tilde{H}^\dagger$ is the left inverse of $\tilde{H}$ from \eqref{eqn:dlmpc_admm}; $\tilde{H}$ has full column rank. When formulating the DLMPC problem \eqref{eqn:dlmpc} as an ADMM problem, we perform variable duplication to obtain problem \eqref{eqn:dlmpc_admm}. We can write \eqref{eqn:dlmpc_admm} in terms of $h(\mathbf{\tilde\Phi})$ with the constraint $\mathbf{\tilde\Phi}=\tilde H\mathbf{\tilde\Psi}$:
\begin{equation*}
\begin{aligned}
& \underset{\mathbf{\tilde\Phi},\mathbf{\tilde\Psi}}{\text{min}} && h(\mathbf{\tilde\Phi}) \ \text{s.t.} &    &\mathbf{\tilde\Phi}=\tilde H\mathbf{\tilde\Psi}.
\end{aligned}
\end{equation*}

By assumption, $f(M_1\mathbf{\tilde\Phi}x_0)$ is closed, proper, and convex, and $\mathcal{\tilde P}$ is a closed and convex set. The remaining constraints $Z_{AB}\mathbf{\tilde\Phi}=I$ and $\mathbf{\tilde\Phi}\in\mathcal{L}_{d}$ are also closed and convex. Hence, $h(\mathbf{\tilde\Phi})$ is closed, proper, and convex.

We now show \textit{2)}. This condition is equivalent to showing that strong duality holds \cite{boyd_convex_2004}. Since problem \eqref{eqn:dlmpc} is assumed to have a feasible solution in the relative interior of $\mathcal{P}$ by means of Lemma \ref{lemm:terminal_set} (given that the first iteration is feasible), 
Slater's condition is automatically satisfied, and therefore the unaugmented Lagrangian of the problem has a saddle point.

Both conditions of the ADMM convergence result from \cite{boyd_distributed_2010} are satisfied -- Algorithm \ref{alg:dlmpc} converges in residue, objective function, and dual variable, as desired.

\end{proof}

\section{Feasible and Stable DLMPC}\label{sec:implementation}

We incorporate theoretical results from \S \ref{sec:guarantees} into the DLMPC computation. We provide an algorithm to compute the terminal set $\mathcal X_T$. We also provide a distributed and localized computation for the terminal cost function $f_T$. The terminal set and cost generally introduce local coupling among subsystems; this minimizes conservatism but requires an extension to Algorithm \ref{alg:dlmpc} to accommodate coupling. All algorithms provided are distributed and localized, with computational complexity that is independent of the global system size.

\subsection{Offline synthesis of the terminal set $\mathcal X_T$}

Our first result is to provide an offline algorithm to compute terminal set $\mathcal X_T$ from \eqref{eqn:terminal_set} in a distributed and localized manner. As discussed in \S \ref{sec:guarantees}, we compute the maximal robust positive invariant set for the unconstrained localized closed-loop system. We use SLS-based techniques to obtain a localized closed-loop map $\mathbf{\Phi}$.

Our algorithm is based on Algorithm 10.4 from \cite{borrelli_predictive_2017}. The advantage of implementing this algorithm in terms of the localized closed-loop map $\mathbf{\Phi}$ is twofold: i) we can work with locally-bounded and polytopic disturbance sets $\mathcal{W}$ by leveraging Lemmas 1 and 2 in \cite{amoalonso_implementation_2021}, and ii) the resulting robust invariant set is automatically localized given the localized structure of the closed-loop map.

We start by finding a localized closed-loop map $\mathbf{\Phi}$ for system \eqref{eqn:system}. To do this, we need to solve
\begin{equation} \label{eqn:sls_unconstrained}
\underset{\mathbf{\Phi}}{\text{min }} f(\mathbf{\Phi}) \quad \text{~s.t. } Z_{AB}\mathbf{\Phi}=I, ~ \mathbf{\Phi}\in\mathcal L_d.
\end{equation}
This is an SLS problem with a separable structure for most standard cost functions $f$ \cite{anderson_system_2019}. The separable structure admits localized and distributed computation. Even when separability is not apparent, a relaxation can often be found that allows for distributed computation (see for example \cite{amo_alonso_distributedLQR_2020,wang_large-scale_2021}). The infinite-horizon solution for quadratic cost is presented in \cite{yu_localized_2020}. For other costs, computing an infinite horizon solution to \eqref{eqn:sls_unconstrained} remains an open question -- in these cases, we can use finite impulse response SLS with sufficiently long time horizon.

Once a localized closed-map $\mathbf{\Phi}$ has been found, we can compute the associated maximal positive invariant set. In the robust case, we leverage results from Lemmas 1 and 2 in our companion paper \cite{amoalonso_implementation_2021}, which use duality arguments to tackle specific formulations of $\mathcal{W}$. We denote $\sigma$ as the upper bound of $\left\Vert [\mathbf\delta]_i \right\Vert_*$ for all $i$, where $\left\Vert \cdot \right\Vert_*$ is the dual norm of $\left\Vert \cdot \right\Vert_p$. Also, Each $e_j$ is the $j^{th}$ vector in the standard basis.
\begin{itemize}
\item Nominal (i.e. no disturbance):
\begin{align*}
\hspace{-6mm}\mathcal S:=&\{x_0\in\mathbb{R}^n:\ \mathbf\Phi\{1\} x_0\in\mathcal P\}.
\end{align*}
\item Locally bounded disturbance:
$$\mathcal S:=\{x_0\in\mathbb{R}^n:\ [H]_i [\mathbf\Phi\{1\}]_i [x_0]_i\ + 
\sum_j\sigma \left\Vert e_j^\intercal [H]_i [\mathbf\Phi\{2:T\}]_i \right\Vert_* \leq[h]_i\; \forall i\}.$$
\item Polytopic disturbance:
\begin{multline*}
\mathcal S:=\{x_0\in\mathbb{R}^n:\ H\mathbf\Phi\{1\}x_0+\Xi g \leq h, \text{where }\Xi(j,:) =\ \underset{\Xi_j\geq0}{\text{min}}\ \Xi_j g \ \text{s.t.}\ H(j,:)\mathbf\Phi\{2:T\}=\Xi_j G\ \forall j\}.
\end{multline*}
\end{itemize}

As per Lemma \ref{lemm:terminal_set}, we calculate $\mathcal{S}$ by iteratively computing $\mathcal{S}_{k+1}$ from $\mathcal{S}_k$.\footnote{For simplicity of presentation, we write out $\mathcal{S}$ only for finite-horizon $\mathbf{\Phi}$; the proposed algorithm to synthesize $\mathcal{S}$ works for infinite-horizon $\mathbf{\Phi}$ as well.} Assume $\mathcal{S}_k$ can be written as $$\mathcal{S}_k = \{x_0 \in \mathbb{R}^n:\ \hat{H}x \leq \hat{h} \}. $$

Then, we can write $\mathcal{S}_{k+1}$ as
\begin{equation} \label{eqn:precursor_set}
\begin{aligned} 
\mathcal{S}_{k+1} = \{x_0 \in \mathbb{R}^n:\ & \Phi_{x,1}[1]x_0 \in \mathcal{F}(\mathcal{S}_k), \Phi_{u,0}[0]x_0 \in \mathcal{U} \},
\end{aligned}
\end{equation}
where $\mathcal{F}(\mathcal{S}_k) := \mathcal{S}_k$ in the nominal case. In the case of locally bounded disturbance,
\begin{equation*}
\begin{aligned}
\mathcal{F}(\mathcal{S}_k) := \{& x \in \mathbb{R}^n:\ [\hat{H}]_i[x]_i\ + \sum_j\sigma \left\Vert e_j^\intercal [\hat{H}]_i \right\Vert_* \leq[\hat{h}]_i\; \forall i\},
\end{aligned}
\end{equation*}
and for polytopic disturbance,
$$\mathcal{F}(\mathcal{S}_k) := \{x\in \mathbb{R}^n:\ \hat{H}x + \Xi g\leq \hat{h}, \text{where } \Xi(j,:) =\ \underset{\Xi_j\geq0}{\text{min}}\ \Xi_j g \ \text{s.t.}\ \hat{H}(j,:)=\Xi_j G\ \forall j\}.$$

These formulations use the simplifying fact that $\Phi_{x,0}[0] = I$ for all feasible closed-loop dynamics. Also, notice that by using $\mathcal{S}_k$ to calculate $\mathcal{S}_{k+1}$, the only elements of $\mathbf{\Phi}$ that we require are $\Phi_{x,1}[1]$ and $\Phi_{u,0}[0]$.

The conditions stated in \eqref{eqn:precursor_set} are row- and column-wise separable in $\mathbf \Phi$ (and $\Xi$)\footnote{A formal definition or row and column-wise separability is given in \cite{amoalonso_implementation_2021}.}; this allows us to calculate $\mathcal{S}_k$ in a distributed way. Furthermore, $\mathcal{F}(\mathcal{S}_k)$ and $\mathcal{U}$ are localizable by Assumption \ref{assump:locality}, in both nominal and robust settings. Since $\mathbf{\Phi}$ is also localized, we can exploit the structure of $\mathbf \Phi$ to and rewrite $S_{k+1}$ as the intersection of local sets, i.e. $S_{k+1} = S_{k+1}^1 \cap \dots \cap S_{k+1}^N$, where
\begin{equation}\label{eqn:precursor_local}
\begin{aligned} 
\mathcal{S}_{k+1}^i = \{& [x_0]_{\mathbf{in}_i(d)} \in \mathbb R^{[n]_i}:\ \ [\Phi_{x,1}[1]]_{\mathbf{in}_i(d)}[x_0]_{\mathbf{in}_i(d)} \in \mathcal{F}(\mathcal{S}_k^{\mathbf{in}_i(d)}), \ 
[\Phi_{x,1}[1]]_{\mathbf{in}_i(d)}[x_0]_{\mathbf{in}_i(d)} \in \mathcal{U}^{\mathbf{in}_i(d)} \}
\end{aligned}
\end{equation}

We now present Algorithm \ref{alg:terminal_set} to compute the terminal set $\mathcal X_{T} := \mathcal{S}$ in a distributed and local manner. Each subsystem $i$ computes its own local terminal set $\mathcal{S}^i$ using only local information exchange. This algorithm is inspired by Algorithm 10.4 in \cite{borrelli_predictive_2017}.

\setlength{\textfloatsep}{0pt}
		\begin{algorithm}[ht]
		\caption{Subsystem $i$ terminal set computation}\label{alg:terminal_set}
		\begin{algorithmic}[1]
		\Statex \textbf{input:} $[\Phi_{x,1}[1]]_{\mathbf{in}_i(d)}, [\Phi_{u,0}[0]]_{\mathbf{in}_i(d)},\mathcal{U}^{\mathbf{in}_i(d)}$
		\Statex \quad for polytopic noise, also $[G]_{\mathbf{in}_i(d)}, [g]_{\mathbf{in}_i(d)}$
		\State $\mathcal{S}_0^i \leftarrow \mathcal{X}^i, k\leftarrow -1$
		\State \textbf{repeat:} 
		\State \quad $k\leftarrow k+1$
		\State $\quad$Share $\mathcal S_{k}^i $ with $\textbf{out}_{i}(d)$ .
		\Statex $\quad$Receive $\mathcal S_{k}^j$ from all $j\in\textbf{in}_{i}(d)$.
		\State $\quad$Compute $\mathcal S_{k+1}^i$ via \eqref{eqn:precursor_local}.
		\State $\quad\mathcal S_{k+1}^i \leftarrow \mathcal S_{k}^i \cap \mathcal S_{k+1}^i $
		\State \textbf{until:} $\mathcal S_{k+1}^i  = \mathcal S_{k}^{i}$ for all $i$
		\State $\mathcal{X}_{T}^{i} \leftarrow \mathcal S_{k}^i $ 
		\Statex \textbf{output:} $\mathcal X^i_T$
		\end{algorithmic}
		\end{algorithm}

If state and input constraints $\mathcal{X}$ and $\mathcal{U}$ do not induce coupling (as assumed in Algorithm \ref{alg:dlmpc}), the resulting terminal set $\mathcal{X}_T$ will be at most $d$-localized. If $\mathcal{X}$ and $\mathcal{U}$ do induce $d$-local coupling, the terminal set will be at most $2d$-localized since -- in the presence of coupling --  Alg. \ref{alg:terminal_set} requieres communication with not only local patch neighbors, but also neighbors of those neighbors. Convergence is guaranteed for system \eqref{eqn:system} if it is stable when $u=0,\ w=0$, and when the constraint and disturbance sets $\mathcal X,\ \mathcal U,\ \mathcal W$ are bounded and contain the origin, as per \cite{gilbert_linear_1991}.

\subsection{Online synthesis of DLMPC}

DLMPC Algorithm \ref{alg:dlmpc} was developed under the assumption that no coupling is introduced through cost or constraints. We now extend this algorithm to accommodate local coupling, which is allowed as per Assumption \ref{assump:locality}. This allows us to incorporate the terminal set and cost introduced in the previous sections, both of which generally induce local coupling. We will use notation that assumes all constraints and costs (including the terminal set) are $d$-localized. In the case that the terminal set is $2d$-localized, subsystems will need to exchange information with neighbors up to $2d$-hops away; simply replace \{ $\textbf{out}_{i}(d)$, $\textbf{in}_{i}(d)$ \} with \{ $\textbf{out}_{i}(2d)$, $\textbf{in}_{i}(2d)$ \} wherever they appear in Algorithms \ref{alg:dlmpc}, \ref{alg:terminal_set} and \ref{alg:consensus}.

Appending the terminal set \eqref{eqn:terminal_set} and terminal cost \eqref{eqn:terminal_cost} to the DLMPC problem \eqref{eqn:dlmpc_admm}, gives:
\begin{align}
&\underset{\mathbf{\tilde\Phi},\mathbf{\tilde\Psi},\eta}{\text{min}} & f(M_1\mathbf{\tilde \Phi}\{1\}x_{0}) + \eta  \label{eqn:dlmpc_terminal} & ~\\
& ~~~ \text{s.t.} & Z_{AB}M_2\mathbf{\tilde \Psi}=I, x_0 =& x(\tau), ~ \mathbf{\tilde\Phi},\mathbf{\tilde \Psi}\in\mathcal{L}_d, \nonumber \\
& ~ &\mathbf{\tilde\Phi}x_0\in\mathcal{\tilde P},\ \tilde H\mathbf{\tilde \Phi}=&\mathbf{\tilde\Psi}, \nonumber \\
& & M_1\mathbf{\tilde \Phi}_T\{1\}x_{0} \in\eta\mathcal X_T, & ~ 0\leq\eta\leq 1 \nonumber,
\end{align}
\noindent where $\mathbf{\tilde \Phi}_T$ represents block rows of $\mathbf{\tilde \Phi}$ that correspond to time horizon $T$ -- e.g., in the nominal setting, the last block row of $\mathbf{\Phi}_x$.

\begin{remark} In the robust setting, we incorporate the terminal constraint into $\tilde{\mathcal P}$ to ensure robust satisfaction of the terminal constraint. However, $\mathcal X_T$ still appears as nominal constraint in order to integrate the definition of the terminal cost \eqref{eqn:terminal_cost} into the DLMPC formulation \eqref{eqn:dlmpc}.
\end{remark}

In the original formulation \eqref{eqn:dlmpc_admm}, we assume no coupling; as a result, all expressions involving $\mathbf{\tilde\Phi}$ are row-separable, and all expressions involving $\mathbf{\tilde\Psi}$ are column-separable. If we allow local coupling as per Assumption \ref{assump:locality}, we lose row-separability; row-separability is also lost in the new formulation \eqref{eqn:dlmpc_terminal} due to local coupling in the terminal set. This is because without coupling, $[x]_i$ and $[u]_i$ (which correspond to a fixed set of rows in $\mathbf{\Phi}$), can only appear in $f^i$ and $\mathcal P^i$; thus, each row of $\mathbf{\Phi}$ (and $\mathbf{\tilde\Phi}$) is solved by exactly one subsystem in step 3 of Algorithm \ref{alg:dlmpc}. With coupling, $[x]_i$ and $[u]_i$ can appear in $f^j$ and $\mathcal P^j$ for any $j\in\textbf{in}_i(d)$; now, each row of $\mathbf{\Phi}$ is solved for by multiple local subsystems at once, and row-separability is lost. This only affects step 3 of Algorithm \ref{alg:dlmpc} (i.e. the row-wise problem); other steps remain unchanged. We write out the row-wise problem corresponding to \eqref{eqn:dlmpc_terminal}:
\begin{align}
& \underset{\mathbf{\tilde\Phi}}{\text{min}} & f(M_1\mathbf{\tilde \Phi}\{1\}x_{0}) + \eta +  g&( \mathbf{\tilde\Phi},\mathbf{\tilde\Psi}^k,\mathbf{\Lambda}^k)    \label{eqn:dlmpc_Phi} \\
& ~ \text{s.t.} & \mathbf{\tilde\Phi}\in\mathcal{\tilde P}\cap \mathcal L_d,\ M_1\mathbf{\tilde \Phi}_T\{1&\}x_{0} \in\eta\mathcal X_T, \nonumber \\
& & ~0\leq\eta\leq1, ~ x_0 = \ x(\tau), & \nonumber
\end{align}
where $ g( \mathbf{\tilde\Phi},\mathbf{\tilde\Psi},\mathbf{\Lambda}) = \frac{\rho}{2}\left\Vert \mathbf{\tilde\Phi} - \tilde H \mathbf{\tilde\Psi} + \mathbf{\Lambda} \right\Vert_F^2,$
and all other relevant variables and sets are defined as per equation (10) in our companion paper \cite{amoalonso_implementation_2021}.

Note that if $\mathcal P$ induces coupled linear inequality constraints, row-separability is maintained. Though $[x]_i$ and $[u]_i$ appear in multiple $\mathcal{P}_j$, this corresponds to distinct rows of $\mathbf\Omega$, each solved by one subsystem, as opposed to one row of $\mathbf\Phi$ that is solved for by multiple subsystems. A similar idea applies if coupling is induced by some quadratic cost, i.e. $f(\mathbf{\Phi}x_0) = \|C\mathbf{\Phi}x_0\|_F^2$ for some matrix $C$ which induces coupling. In this case, we can modify $\tilde H$ such that we enforce $\mathbf{\Phi} = C\mathbf{\mathbf{\Psi}}$, and rewrite the cost as $\|\mathbf{\Phi}x_0\|_F^2$; now, each row of $\mathbf\Phi$ is solved by exactly one subsystem, and row-separability is maintained.\footnote{We would also need to use $[\Psi_{u,0}[0]]_{i}$ for the algorithm output instead of $[\Phi_{u,0}[0]]_{i}$.} However, this technique does not apply to coupling in the general case; nor does it apply to the coupling induced by the terminal cost -- it is generally true that each row of $\mathbf{\tilde \Phi}$ will be need to be solved for by several subsystems. We introduce a new vector variable
$$\mathbf X := M_1\mathbf{\tilde \Phi}\{1\}x_{0}.$$
This variable facilitates consensus between subsystems who share the same row(s) of $\mathbf{\tilde \Phi}$. Each subsystem solves for components $[\mathbf X]_{\mathbf{in}_i(d)}$, and comes to a consensus with its neighboring subsystems on the value of these components. In the interest of efficiency, we directly enforce consensus on elements of $\mathbf X$ instead of enforcing consensus on rows of $\mathbf{\tilde \Phi}$. We introduce a similar variable for terminal cost $\eta$; we define vector $\boldsymbol \eta := \begin{bmatrix}\eta_1,\dots,\eta_N\end{bmatrix}$. Subsystem $i$ solves for $\eta_i$, which is its own copy of $\eta$. It then comes to a consensus on this value with its neighbors, i.e. $\eta_j$ has the same value $\forall j \in\textbf{in}_i(d)$. Assuming that $\mathcal{G}_{(A,B)}$ is connected, this guarantees that $\eta_i$ has the same value $\forall i \in \{1 \dots N \}$, i.e., all subsystems agree on the value of $\eta$. We combine these two consensus-facilitating variables into the augmented vector variable
$$\mathbf{\tilde X} := \begin{bmatrix} \mathbf X^\intercal & \boldsymbol \eta^\intercal \end{bmatrix}^\intercal.$$
With this setup, we follow a variable duplication strategy and apply ADMM for consensus \cite{costantini_decomposition_2018}. In particular, we duplicate variables so each subsystem has \emph{its own copy} of the components of $\mathbf{\tilde X}$. Problem \eqref{eqn:dlmpc_Phi} becomes:
\begin{align}
& \underset{\mathbf{\tilde\Phi},\mathbf{\tilde X},\mathbf{\tilde Y}}{\text{min}} & \tilde f(\mathbf{\tilde X}) + g( \mathbf{\tilde\Phi},\mathbf{\tilde\Psi}^k,\mathbf{\Lambda}^k )  \label{eqn:dlmpc_X} \\
&~~ \text{s.t.} & \mathbf{\tilde\Phi},\mathbf{\tilde X}\in\mathcal{Q},\ \mathbf{\tilde X} \in\mathcal {\tilde X}_T, ~&\mathbf{\tilde\Phi}\in\mathcal{L}_d, ~x_0 = \ x(\tau), \nonumber \\
& & [\mathbf{\tilde{X}}]_i=[M_1\mathbf{\tilde\Phi}]_{i_r}[x_0]_i,~ &[\mathbf{\tilde{X}}]_j=[\mathbf{\tilde{Y}}]_j\ \forall j\in\mathbf{in}_i(d) \ \forall i, \nonumber
\end{align}
where we define $\tilde f$, $\mathcal {\tilde X}_T$ and $\mathcal Q$ as follows:
\begin{itemize}
\item $\tilde f(\mathbf{\tilde X}) := f(\mathbf X)+\frac{1}{N}\sum_{i=1}^N \eta_i,$
\item $\mathbf{\tilde X} \in\mathcal {\tilde X}_T$ if and only if $M_1\mathbf{\tilde \Phi}_T\{1\}x_{0} \in\eta_i\mathcal X_T\ \forall i$,
\item $\mathbf{\tilde\Phi},\mathbf{\tilde X}\in\mathcal{Q}$ if and only if $\mathbf 0 \leq \boldsymbol{\eta} \leq \mathbf 1$ \textbf{and} $\{ \mathbf{\tilde\Phi}\in\mathcal{\tilde P}$ if $\mathcal{\tilde P}$ induces linear inequalities \textbf{or} $\mathbf{\tilde X}\in\mathcal P$ otherwise$\}$.\footnote{By the assumptions in \cite{amoalonso_implementation_2021}, $\mathcal P$ (and therefore $\mathcal{\tilde{P}}$) must induce linear inequalities in the robust setting. The only case where  $\mathcal{\tilde{P}}$ may induce something other than linear inequalities is in the nominal setting.}
\end{itemize}

The structure of problem \eqref{eqn:dlmpc_X} allows us to solve it in a distributed and localized manner via ADMM-based consensus. In particular, each subsystem $i$ solves:
\begin{subequations}\label{eqn:consensus}
\begin{align}
	&\hspace{-2mm}\left\{\hspace{-2mm} \begin{array}{c} [\mathbf{\tilde\Phi}]_{i_{r},}^{k+1,n+1}\\ ~\\{} [\mathbf{\tilde X}]^{n+1}_{\mathbf{in}_i(d)} \end{array} \hspace{-2mm}\right\} \hspace{-1mm}=
			\hspace{-1mm}\left\{\begin{aligned}
			& \hspace{-7mm} \underset{ ~~~~~ \scriptscriptstyle{[\mathbf{\tilde \Phi}]_{i_r},[\mathbf{\tilde X}]_{\mathbf{in}_i(d)}}}{\text{argmin}} && \hspace{-6mm}\tilde g^{i}_k(\mathbf{\tilde X},\mathbf{\tilde \Phi}) \hspace{-1mm}+ \hspace{-1mm} \frac{\mu}{2}h^{i}(\mathbf{\tilde X},\mathbf{\tilde Y}^n,\mathbf{\tilde Z}^n)\\
			&~~~ \text{s.t.} &&  \hspace{-6mm}
				\begin{aligned}
     			 		  & [\mathbf{\tilde\Phi}]_{r_i}\hspace{-1mm}\in\hspace{-1mm}\mathcal{Q}^i,\ [\mathbf{\tilde X}_T]_i\hspace{-1mm} \in\hspace{-1mm}\mathcal {\tilde X}_T^i\cap \mathcal{Q}^i, \\
    					  & [\mathbf{\tilde{X}}]_i=[M_1\mathbf{\tilde\Phi}]_{r_i}[x_0]_i
			     	 \end{aligned}
		\end{aligned} \right\}\label{eqn:consensus_X}\\[-4pt]
	&[\mathbf{\tilde Y}]_{i}^{n+1} = \frac{h^{i}(\mathbf{\tilde X}^{n+1},\mathbf{0},\mathbf{\tilde Z}^n)}{\vert\textbf{in}_{i}(d)\vert},\label{eqn:consensus_Y} \\[5pt]
	& \mathbf{[\tilde Z]}_{ij}^{n+1} = \mathbf{[\tilde Z]}_{ij}^{n}+\mathbf{[\tilde X]}_{i}^{n+1}-\mathbf{[\tilde Y]}_{j}^{n+1}, \label{eqn:consensus_Z}
\end{align}
\end{subequations} 
where to simplify notation, we define
\begin{align*}
&\tilde g^{i}_k(\mathbf{\tilde X},\mathbf{\tilde \Phi}) := \tilde f^{i}([\mathbf{\tilde X}]_{\mathbf{in}_i(d)})+g\big([\mathbf{\Phi}]_{i_{r}}-[\mathbf{\Psi}]_{i_{r}}^{k}+[\mathbf{\Lambda}]_{i_{r}}^{k}\big), \\[3pt]
& h^i(\mathbf{\tilde X},\mathbf{\tilde Y}^n,\mathbf{\tilde Z}^n) := {\sum}_{j\in \textbf{in}_{i}(d)}\left\Vert[\mathbf{X}]_{j}-[\mathbf{Y}]_{i}^{n}+[\mathbf{Z}]_{ij}^{n}\right\Vert^{2}_{F},
 \end{align*}
Consensus iterations are denoted by $n$, outer-loop (i.e. Algorithm \ref{alg:dlmpc}) iterations are denoted by $k$, and $\mu$ is the ADMM consensus parameter. Intuitively, $\tilde g^{i}_k$ represents the original objective from \eqref{eqn:dlmpc_Phi}, and $h^i$ represents the consensus objective. 

The subroutine described by \eqref{eqn:consensus} allows us to accommodate local coupling induced by cost and constraint (including terminal), and can be implemented in a distributed and localized manner. As stated above, this subroutine solves the row-wise problem \eqref{eqn:dlmpc_Phi}, corresponding to step 3 of Algorithm \ref{alg:dlmpc}. Thus, in order to accommodate local coupling (including terminal set and cost), we need only to replace step 3 of Algorithm \ref{alg:dlmpc} with the subroutine defined by Algorithm \ref{alg:consensus} below. Convergence is guaranteed by a similar argument to Lemma \ref{lemm:convergence}.
 
\begin{algorithm}[h]
\caption{Subsystem $i$ implementation of step 3 in Algorithm \ref{alg:dlmpc} when subject to localized coupling}\label{alg:consensus}
\begin{algorithmic}[1]
\Statex \textbf{input:} tolerance parameters $\epsilon_x, \epsilon_z, \mu >0$.
\State $n\leftarrow0$. 
\State Solve optimization problem \eqref{eqn:consensus_X}.
\State Share $[\mathbf{X}]_{i}^{n+1}$ with $\textbf{out}_{i}(d)$. Receive the corresponding $[\mathbf{X}]_{j}^{n+1}$ from $j\in\textbf{in}_{i}(d)$.
\State Perform update \eqref{eqn:consensus_Y}.
\State Share $[\mathbf{Y}]_{i}^{n+1}$ with $\textbf{out}_{i}(d)$. Receive the corresponding $[\mathbf{Y}]_{j}^{n+1}$ from $j\in\textbf{in}_{i}(d)$.
\State Perform update \eqref{eqn:consensus_Z}.
\State \textbf{if} $\left\|[\mathbf{X}]^{n+1}_i - [\mathbf{Z}]^{n+1}_i\right\|_F<\epsilon_x$
\Statex and $\left\|[\mathbf{Z}]^{n+1}_i - [\mathbf{Z}]^{n}_i\right\|_F<\epsilon_z$ \textbf{:}
\Statex $\;\;$ Go to step 4 in Algorithm \ref{alg:dlmpc}. 
\Statex \textbf{else:}
\Statex $\;\;$ Set $n\leftarrow n+1$ and return to step 2.
\end{algorithmic}
\end{algorithm}

\textbf{Computational complexity of the algorithm:} 
The presence of coupling induces an increase in computational complexity compared to the uncoupled scenario (i.e., Algorithm \ref{alg:dlmpc}); however, the scalability properties from the uncoupled scenario still apply. Complexity in the current algorithm is determined by steps 2, 4, 6 of Algorithm \ref{alg:consensus} and steps 5 and 7 of Algorithm \ref{alg:dlmpc}. Except for step 2 in Algorithm \ref{alg:consensus}, all other steps can be solved in closed form. Sub-problems solved in step 2 require $O(d^{2}T^2)$ optimization variables in the robust setting ($O(dT^2)$ in the nominal setting) and $O(d^2T)$ constraints. All other steps enjoy less complexity since their evaluation reduces to multiplication of matrices of dimension $O(d^2T^2)$ in the robust setting, and $O(d^2T)$ in the nominal setting. The difference in complexity between the nominal and robust settings is consistent with the uncoupled scenario. Compared to the uncoupled scenario, additional computation burden is incurred by the consensus subroutine in Algorithm \ref{alg:consensus}, which increases the total number of iterations. The consensus subroutine also induces increased communication between subsystems, as it requires local information exchange. However, this exchange is limited to a $d$-local region, resulting in small consensus problems that converge quickly, as we illustrate empirically in \S \ref{sec:simulation}. As with the original uncoupled algorithm, complexity of this new algorithm is determined by the size of the local neighborhood and does not increase with the size of the global network.

\section{Simulation Experiments} \label{sec:simulation}

Using examples, we demonstrate how adding terminal constraint and cost affects the performance of the DLMPC algorithm. We verify that introducing terminal constraint and cost indeed produces the desired feasibility and stability properties. We additionally empirically characterize the computational complexity of algorithms presented in previous sections and verify that complexity is independent of global network size for both offline and online algorithms. Code to replicate these experiments is available at \texttt{\url{https://github.com/unstable-zeros/dl-mpc-sls}}; this code makes use of the SLS toolbox \cite{slstoolbox}, which includes ready-to-use MATLAB implementations of all algorithms presented in this paper and its companion paper \cite{amoalonso_implementation_2021}.

\subsection{System model}
Simulations are performed on the system from our companion paper \cite{amoalonso_implementation_2021}; a two-dimensional square mesh, where each node represents a two-state subsystem that follows linearized and discretized swing dynamics
\begin{equation*}
\begin{bmatrix} \theta(t+1) \\ \omega(t+1) \end{bmatrix}_i = \sum_{j\in\textbf{in}_i(1)}[A]_{ij} \begin{bmatrix} \theta(t) \\ \omega(t) \end{bmatrix}_j + [B]_{i}[u]_i + [w]_i,
\end{equation*}
where $[\theta]_i$, $[\dot{\theta}]_i$, $[u]_i$ are the phase angle deviation, frequency deviation, and control action of the controllable load of bus $i$. The dynamic matrices are 
\[[A]_{ii}=\begin{bmatrix}
   1  & \Delta t \\
  -\frac{k_i}{m_i}\Delta t &  1-\frac{d_i}{m_i}\Delta t
\end{bmatrix}, \ [A]_{ij}=\begin{bmatrix}
   0  & 0 \\
   \frac{k_{ij}}{m_i}\Delta t &  0
\end{bmatrix},\] 
and $[B]_{ii}=\begin{bmatrix} 0 & 1 \end{bmatrix}^\intercal$ for all $i$. 

Connectivity among nodes is determined at random; each node connects to each of its neighbors with a $40\%$ probability. The expected number of edges is $ 0.8*n*(n-1)$. The parameters in bus $i$: $m_i^{-1}$ (inertia inverse), $d_i$ (damping) and $k_{ij}$ (coupling term) are randomly generated and uniformly distributed between $[0,\ 2]$, $[0.5,\ 1]$, and $[1,\ 1.5]$ respectively. We set the discretization step $\Delta t = 0.2$, and define $k_i :=  \sum_{j\in\textbf{in}_i(1)} k_{ij}$.

We study both the nominal setting and robust setting with uniformly distributed polytopic noise. The baseline parameter values are $d=3,\ T=5,\ N=16$ ($4\times4$ grid). Unless otherwise specified, we start with a random-generated initial condition. We use a quadratic cost and polytopic constraints on both angle and frequency deviation, and impose upper and lower bounds.

\subsection{Performance}
The addition of the terminal set and cost to the DLMPC algorithm introduces minimal conservatism in both nominal and robust settings. We study the the DLMPC cost for varying values of the time horizons for three different cases: (i) without terminal set and cost (ii) with terminal set, (iii) with terminal set \emph{and} terminal cost. Results are summarized in Fig. \ref{fig:cost_vs_horizon}. Observe that the difference between the optimal cost across all three cases are negligible, indicating that our proposed terminal set and cost introduce no conservatism (while still providing the necessary theoretical guarantees).

\begin{figure}[htp]
	\centering
	\includegraphics[width=.75\columnwidth]{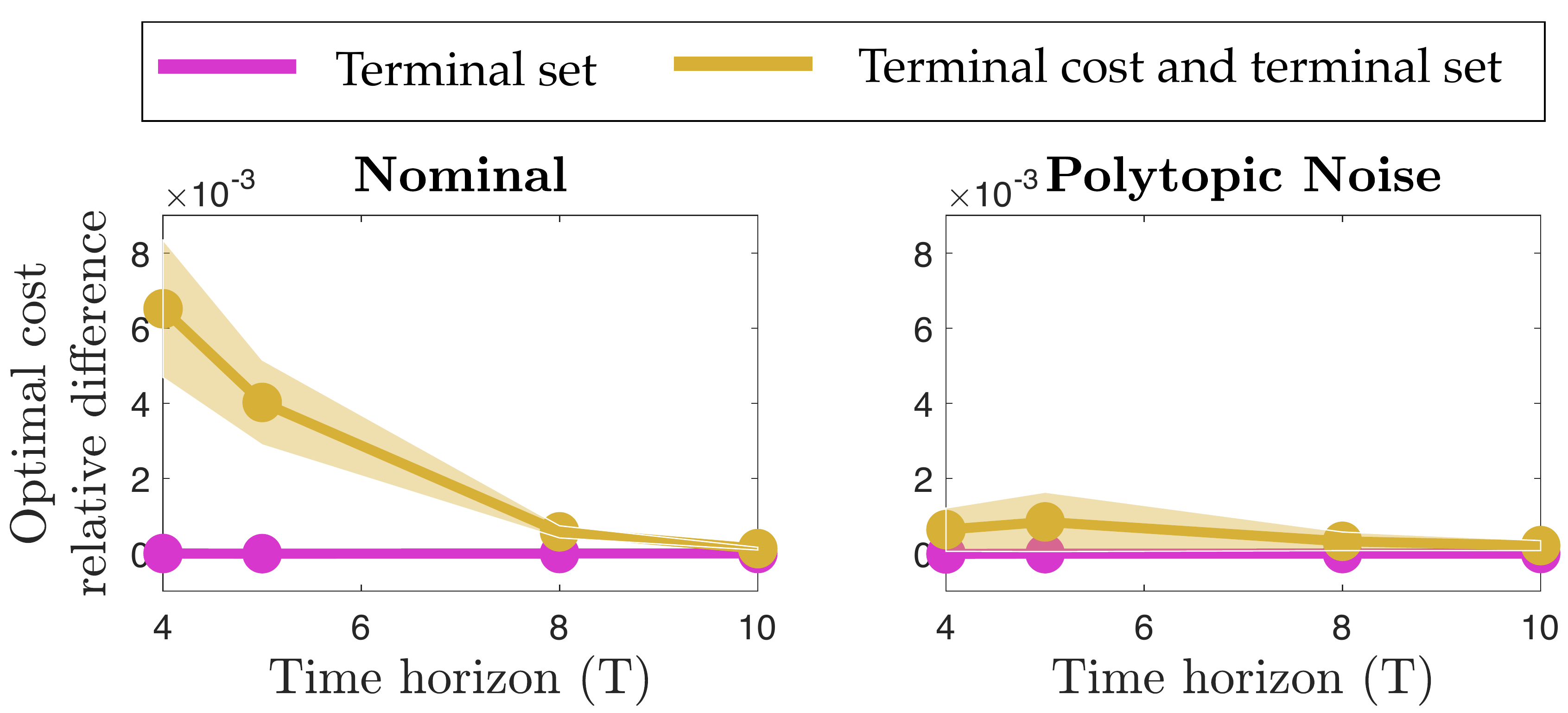}
	\caption{Relative difference of the optimal cost obtained (i) with terminal set (pink) and (ii) with both terminal set and cost (yellow) compared to the optimal cost computed without terminal set and cost. Relative difference is obtained by taking the difference of the two costs and normalizing by the non-terminal-constrained cost. The difference obtained by adding the terminal set (indicated in pink) is on the order of $10^{-5}$ and is not visible on the plot.}
	\label{fig:cost_vs_horizon}
\end{figure}

Generally, the inclusion of the terminal set and cost introduce minimal change. In the vast majority of cases (e.g. all simulations from \cite{amoalonso_implementation_2021}), the DLMPC algorithm is feasible and stable even without a terminal set. This phenomenon has already been observed for the centralized case \cite{borrelli_predictive_2017}. However, we also want to demonstrate how the terminal set and cost \textit{can} make a difference -- example subsystem trajectories for the three different cases are shown in Fig. \ref{fig:dynamics_noiseless} for the nominal case and in the presence of polytopic disturbances. For these simulations only, we use a smaller ($N=5$), more unstable ($m_i^{-1}$ between $[0,\ 16]$) system, extremely short time horizon ($T=2$), and somewhat hand-crafted initial states and disturbances to obtain clearly visible differences between cases -- without such instability, short time horizon, and hand-crafting, differences are generally tiny and not visible. In all cases, the centralized solution (computed via CVX \cite{cvx}) coincides with the solution achieved by the DLMPC Algorithm \ref{alg:dlmpc}, validating the optimality of the proposed algorithm. The effects of introducing terminal set and terminal cost are apparent and consistent with the theoretical results presented in this paper. 

\begin{figure}[htp]
    \centering
   	\includegraphics[width=.75\columnwidth]{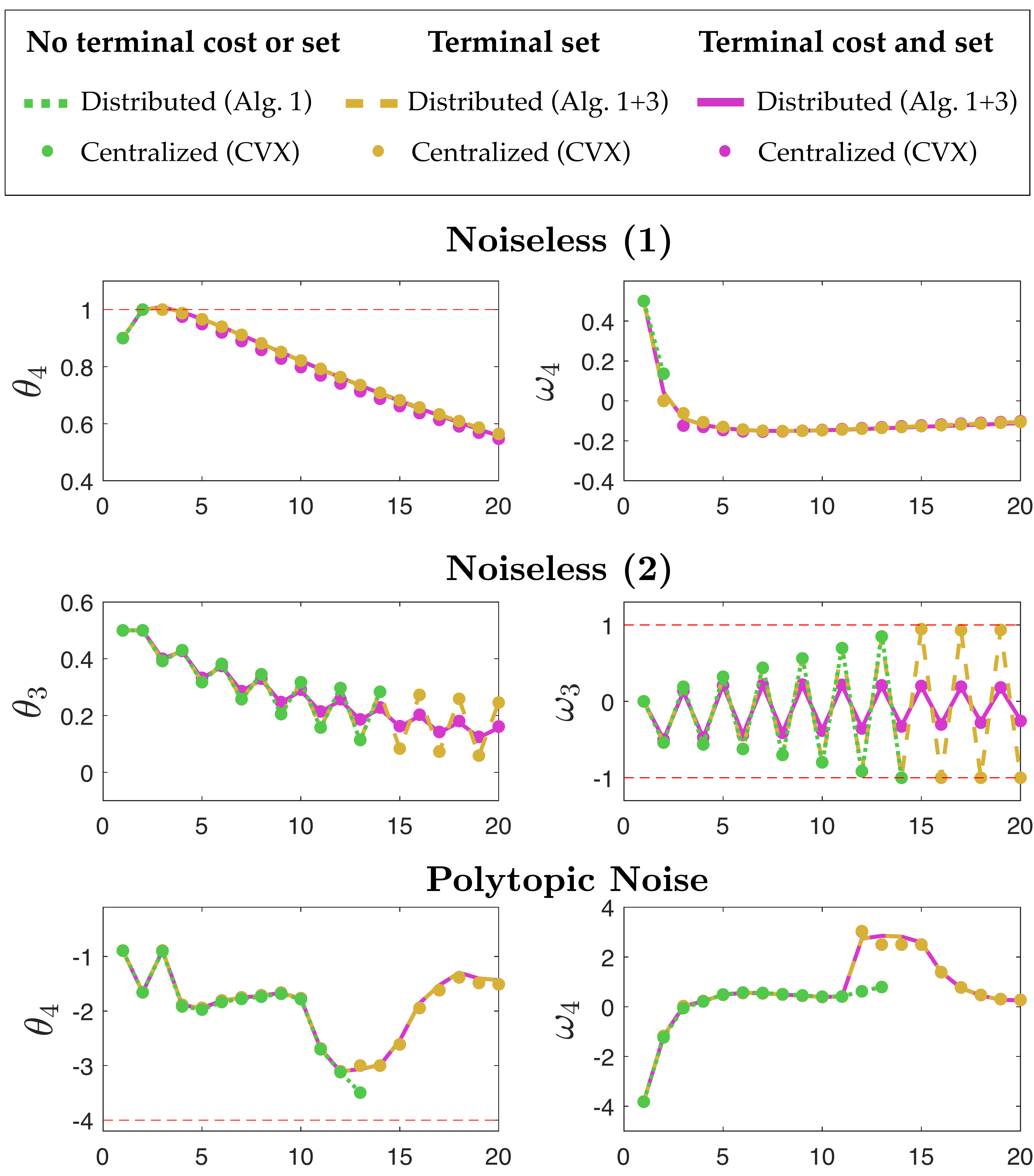}
	\caption{For both nominal cases, upper and lower bounds for all states are 1 and -1 respectively; indicated by red dashed lines in select plots. \textbf{Top (nominal)}: evolution of the states of subsystem $4$ under a DLMPC controller without terminal set and cost (green), with terminal set (pink), and with both terminal set and cost (yellow). Without a terminal set, the algorithm becomes infeasible after $t=2$. With a terminal set, the algorithm remains feasible. Notice the difference in $\omega_4$ for the trajectories at $t=2$. The addition of the terminal cost introduces little change. \textbf{Middle (nominal)}: evolution of the states of subsystem $3$ under different initial conditions. Infeasibility is encountered at $t=15$; here, the difference between the infeasible and feasible trajectories is not visible. Note the oscillations in $\omega_3$; adding a terminal set itself enables feasibility but gives large oscillations -- additionally including the terminal cost results in much smaller, decaying oscillations.
	\textbf{Bottom (polytopic)}: upper and lower bounds for $\theta$ states are -4 and 4, respectively; bounds for $\omega$ states are -20 and 20. Bounds are indicated by red dashed lines in select plots. Without a terminal set, the algorithm becomes infeasible after $t=13$; with a terminal set, feasibility is maintained. Notice the difference in trajectories with and without terminal set just before $t=13$. The addition of the terminal cost introduces little change.}
    \label{fig:dynamics_noiseless}
\end{figure}

\subsection{Computational complexity}
Simulations results verify the scalability of the proposed methods. We measure runtime\footnote{In online simulations, runtime is measured after the first iteration, so that all iterations for which runtime is measured are warm-started.} while varying different network and problem parameters: locality $d$, network size $N$, and time horizon $T$.\footnote{To increase network size, we vary the size of the grid over $4\times4$ ($32$ states), $6\times6$ ($72$ states), $8\times8$ ($128$ states), and $11\times11$ ($242$ states) grid sizes.} We run $5$ different simulations for each of the parameter combinations, using different realizations of the randomly chosen parameters to provide consistent runtime estimates.

First, we study the scalability of the offline Algorithm \ref{alg:terminal_set} to compute the terminal set; results are shown in Fig. \ref{fig:scalability_offline}. Consistent with theoretical analyses in \S \ref{sec:implementation}, runtime does not increase with the size of the network; rather, it increases with the size of the neighborhood. As expected, computations for the robust set take slightly longer than for the nominal set, since the variables in the robust setting have greater dimension. Overall, synthesis times for both nominal and robust settings are extremely low, especially when a small locality size is used.

\begin{figure}[htp]
    \centering
    \includegraphics[width=.7\columnwidth]{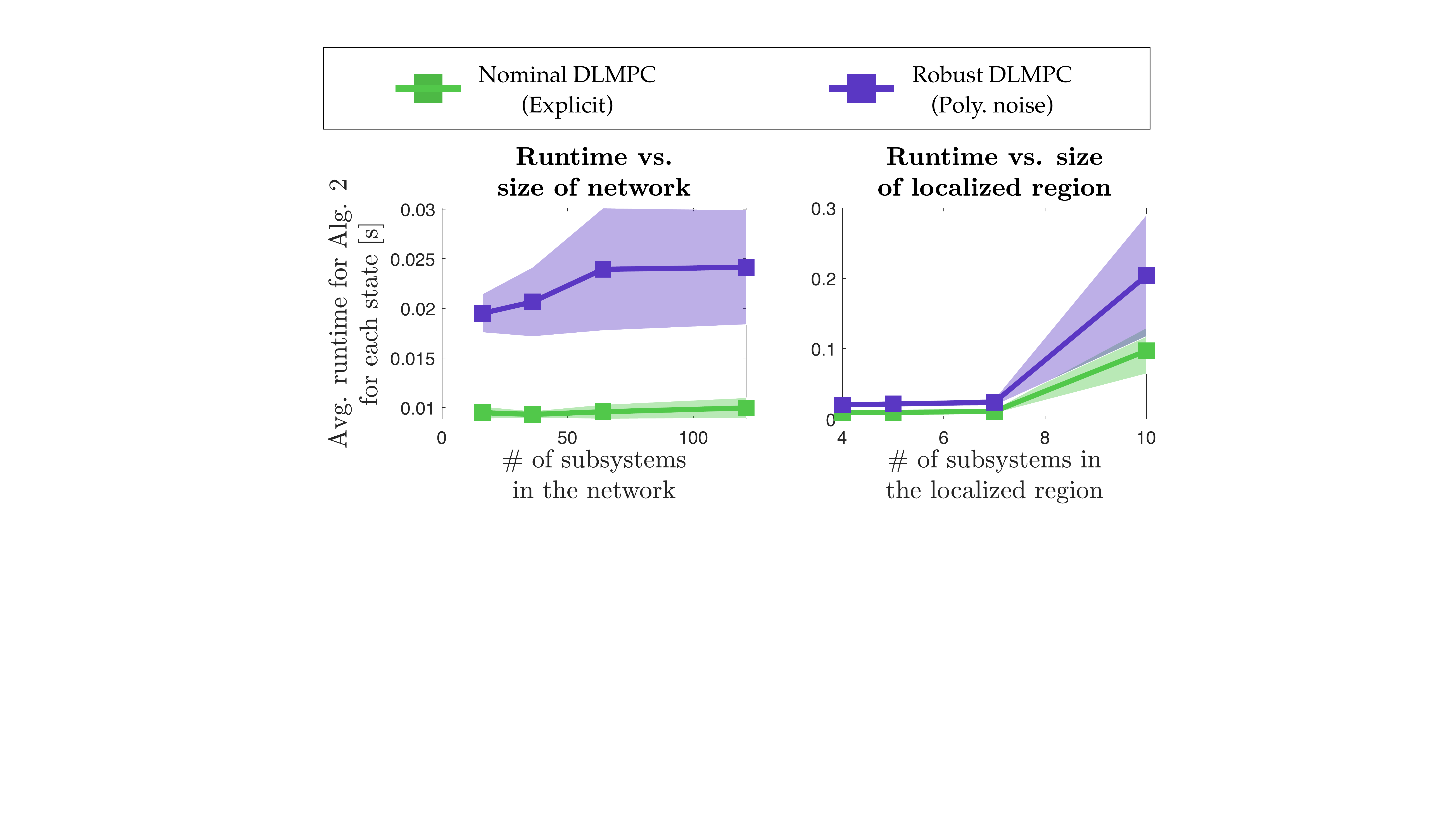}
    \caption{Average runtime of Algorithm \ref{alg:terminal_set} with network size (left) and locality parameter (right). The lines are the mean values and the shaded areas show the values within one standard deviation. Since computations are parallelized across subsystems, runtime is measured on a subsystem, normalized per state, and averaged after the algorithm computation is finished.}
    \label{fig:scalability_offline}
\end{figure}

We also study how scalability of the DLMPC algorithm is affected when we impose a terminal set and terminal cost, and use Algorithm \ref{alg:consensus} to handle coupling. Results are shown in Fig. \ref{fig:scalability_online}; this figure was generated using the same systems and parameters as Fig. 2 from our companion paper \cite{amoalonso_implementation_2021}, allowing for direct comparison of online runtimes. The addition of the terminal set/cost slightly increases runtime, as expected. In the nominal case, runtime is increased from about $10^{-3}$s to $10^{-1}$s. In the case of polytopic disturbances, runtime is increased from about $1-10$s to $10$s. Scalability is maintained; runtime barely increases with the size of the network. Overall, simulations indicate that the introduction of a terminal set and cost preserve scalability, minimally impact computational overhead and performance, and 
provide the desired guarantees.

\begin{figure}[htp]
    \centering
    \includegraphics[width=.5\columnwidth]{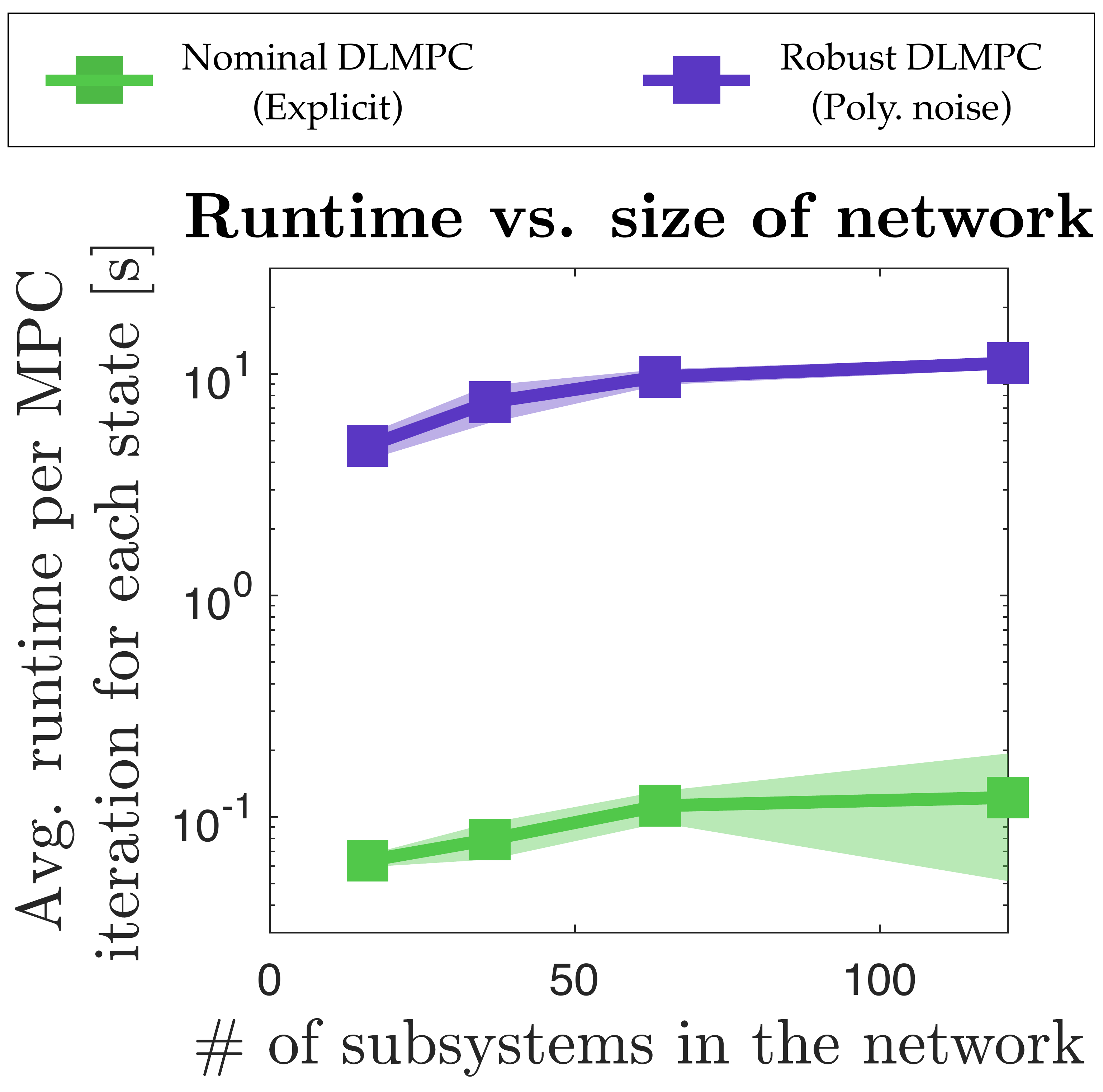}
    \caption{Average runtime per DLMPC iteration with network size when terminal set and terminal cost are imposed. The lines are the mean values and the shaded areas show the values within one standard deviation. Since computations are parallelized across subsystems, runtime is measured on a subsystem, normalized per state, and averaged out after the MPC controller is finished.}
    \label{fig:scalability_online}
\end{figure}

\vspace{3cm}
\section{Conclusion}

In this paper we provide theoretical guarantees for the \emph{closed-loop} DLMPC approach, in both nominal and robust settings. In particular, we ensure recursive feasibility and stability by incorporating a terminal set and terminal cost. We also give guarantees for convergence of the algorithm. For the terminal set, we choose the maximal robust positive invariant set, which can be expressed compactly in the SLS parametrization. We introduce an algorithm to scalably compute this terminal set. We also provide the requisite modifications to the online DLMPC algorithm to accommodate local coupling induced by the terminal set and cost, as well as general coupling induced by process cost and constraints. All algorithms require only local information exchange, and enjoy computational complexity that is independent of the global system size. The results presented in this paper are the first to provide a distributed and localized computation of the maximal robust positive invariant control set and Lyapunov function of a large-scale system. 



\section*{Acknowledgements}
Carmen Amo Alonso is partially supported by a Caltech/Amazon AI4Science fellowship. Jing Shuang (Lisa) Li is supported in part by a postgraduate scholarship from the Natural Sciences and Engineering Research Council of Canada [NSERC PGSD3-557385-2021]. Nikolai Matni is supported in part by NSF awards CPS-2038873 and CAREER award ECCS-2045834, and a Google Research Scholar award. James Anderson is partially supported by NSF CAREER ECCS-2144634 and DOE award DE-SC0022234. 


\bibliographystyle{IEEEtran}
\bibliography{references}

\begin{thebibliography}{10}
\providecommand{\url}[1]{#1}
\csname url@samestyle\endcsname
\providecommand{\newblock}{\relax}
\providecommand{\bibinfo}[2]{#2}
\providecommand{\BIBentrySTDinterwordspacing}{\spaceskip=0pt\relax}
\providecommand{\BIBentryALTinterwordstretchfactor}{4}
\providecommand{\BIBentryALTinterwordspacing}{\spaceskip=\fontdimen2\font plus
\BIBentryALTinterwordstretchfactor\fontdimen3\font minus
  \fontdimen4\font\relax}
\providecommand{\BIBforeignlanguage}[2]{{%
\expandafter\ifx\csname l@#1\endcsname\relax
\typeout{** WARNING: IEEEtran.bst: No hyphenation pattern has been}%
\typeout{** loaded for the language `#1'. Using the pattern for}%
\typeout{** the default language instead.}%
\else
\language=\csname l@#1\endcsname
\fi
#2}}
\providecommand{\BIBdecl}{\relax}
\BIBdecl

\bibitem{borrelli_predictive_2017}
F.~Borrelli, A.~Bemporad, and M.~Morari, \emph{Predictive {Control} for
  {Linear} and {Hybrid} {Systems}}.\hskip 1em plus 0.5em minus 0.4em\relax
  Cambridge University Press, 2017.

\bibitem{mayne_robust_2005}
D.~Q. Mayne, M.~M. Seron, and S.~V. Rakovic, ``Robust model predictive control
  of constrained linear systems with bounded disturbances,'' \emph{Automatica},
  vol.~41, pp. 219 -- 224, 2005.

\bibitem{langbort_distributed_2004}
C.~Langbort, C.~R.S., and D.~R., ``Distributed control design for systems
  interconnected over an arbitrary graph,'' \emph{IEEE Trans. Autom. Control},
  vol.~49, no.~9, pp. 1502 -- 1519, 2004.

\bibitem{jokic_decentralized_2009}
A.~Jokic and M.~Lazar, ``On decentralized stabilization of discrete-time
  nonlinear systems,'' in \emph{Proc. IEEE ACC}, 2009, pp. 5777--5782.

\bibitem{zecevic_control_2010}
A.~I. Zecevic and D.~D. Siljak, \emph{Control of Complex Systems}.\hskip 1em
  plus 0.5em minus 0.4em\relax New York: Commun. Control Eng., Springer, 1988.

\bibitem{stewart_cooperative_2010}
B.~T. Stewart, A.~Venkat, J.~Rawlings, S.~Wright, and G.~Pannocchia,
  ``Cooperative distributed model predictive control,'' \emph{Syst. Control
  Lett.}, vol.~59, no.~8, pp. 460 -- 469, 2010.

\bibitem{maestre_distributed_2011}
J.~M. Maestre, D.~Mu{\~{n}}oz de~la Pe{\~{n}}a, E.~F. Camacho, and T.~Alamo,
  ``Distributed model predictive control based on agent negotiation,'' \emph{J.
  Process Control}, vol.~21, no.~5, pp. 685 -- 697, 2011.

\bibitem{conte_distributed_2016}
C.~Conte, C.~N. Jones, M.~Morari, and M.~N. Zeilinger, ``Distributed synthesis
  and stability of cooperative distributed model predictive control for linear
  systems,'' \emph{Automatica}, vol.~69, pp. 117--125, Jul 2016.

\bibitem{trodden_distributed_2017}
P.~A. Trodden and J.~M. Maestre, ``Distributed predictive control with
  minimization of mutual disturbances,'' \emph{Automatica}, vol.~77, pp. 31 --
  43, 2017.

\bibitem{darivianakis_distributed_2020}
G.~Darivianakis, A.~Eichler, and J.~Lygeros, ``Distributed model predictive
  control for linear systems with adaptive terminal sets,'' \emph{IEEE Trans.
  Autom. Control}, vol.~65, no.~3, pp. 1044 -- 1056, 2020.

\bibitem{aboudonia_distributed_2020}
\BIBentryALTinterwordspacing
A.~Aboudonia, A.~Eichler, and J.~Lygeros, ``Distributed model predictive
  control with asymmetric adaptive terminal sets for the regulation of
  large-scale systems,'' 2020. [Online]. Available:
  \url{https://arxiv.org/abs/2005.04077}
\BIBentrySTDinterwordspacing

\bibitem{muntwiler_distributed_2020}
S.~Muntwiler, K.~P. Wabersich, A.~Carron, and M.~N. Zeilinger, ``Distributed
  model predictive safety certification for learning-based control,''
  \emph{IFAC-PapersOnLine}, vol.~53, no.~2, pp. 5258 -- 5265, 2020.

\bibitem{wang_robust_2021}
\BIBentryALTinterwordspacing
Y.~Wang and C.~Manzie, ``Robust distributed model predictive control of linear
  systems: analysis and synthesis,'' 2021. [Online]. Available:
  \url{https://arxiv.org/abs/2005.04006}
\BIBentrySTDinterwordspacing

\bibitem{sturz_distributed_2020}
Y.~R. Sturz, E.~L. Zhu, U.~Rosolia, K.~H. Johansson, and F.~Borrelli,
  ``Distributed learning model predictive control for linear systems,'' in
  \emph{Proc. IEEE CDC}, 2020, pp. 4366--4373.

\bibitem{amoalonso_implementation_2021}
\BIBentryALTinterwordspacing
C.~{Amo Alonso}, {J. S. Li}, N.~{Matni}, and J.~{Anderson}, ``Distributed and
  localized model predictive control. {Part I}: Synthesis and implementation,''
  2021. [Online]. Available: \url{https://arxiv.org/abs/2110.07010}
\BIBentrySTDinterwordspacing

\bibitem{wang_system_2019}
Y.-S. Wang, N.~Matni, and J.~C. Doyle, ``A system-level approach to controller
  synthesis,'' \emph{IEEE Trans. Autom. Control}, vol.~64, no.~10, pp.
  4079--4093, 2019.

\bibitem{anderson_system_2019}
J.~Anderson, J.~C. Doyle, S.~H. Low, and N.~Matni, ``System level synthesis,''
  \emph{Annu. Rev. Control}, vol.~47, pp. 364 -- 393, 2019.

\bibitem{blanchini_set_1999}
F.~Blanchini, ``Set invariance in control,'' \emph{Automatica}, vol.~35,
  no.~11, pp. 1747--1767, 1999.

\bibitem{wang_separable_2018}
Y.~Wang, N.~Matni, and J.~C. Doyle, ``Separable and localized system-level
  synthesis for large-scale systems,'' \emph{IEEE Trans. Autom. Control},
  vol.~63, no.~12, pp. 4234--4249, Dec. 2018.

\bibitem{boyd_distributed_2010}
S.~Boyd, N.~Parikh, E.~Chu, B.~Peleato, and J.~Eckstein,
  ``\BIBforeignlanguage{en}{Distributed {Optimization} and {Statistical}
  {Learning} via the {Alternating} {Direction} {Method} of {Multipliers}},''
  \emph{\BIBforeignlanguage{en}{Foundations and Trends® in Machine Learning}},
  vol.~3, no.~1, pp. 1--122, 2010.

\bibitem{amo_alonso_distributed_2020}
C.~{Amo Alonso} and N.~{Matni}, ``{Distributed} and localized closed-loop model
  predictive control via {System} {Level} {Synthesis},'' in \emph{Proc. IEEE
  CDC}, 2020, pp. 5598--5605.

\bibitem{amo_alonso_robust_2021}
\BIBentryALTinterwordspacing
C.~{Amo Alonso}, {J.S. Li}, N.~{Matni}, and J.~{Anderson}, ``{Robust}
  distributed and localized model predictive control,'' 2021. [Online].
  Available: \url{https://arxiv.org/abs/2103.14171}
\BIBentrySTDinterwordspacing

\bibitem{mayne_constrained_2000}
D.~Mayne, J.~Rawlings, C.~Rao, and P.~Scokaert, ``Constrained model predictive
  control: Stability and optimality,'' \emph{Automatica}, vol.~36, no.~6, pp.
  789--814, 2000.

\bibitem{lofberg_minimax_2003}
J.~L\"ofberg, ``Minimax approaches to robust modelpredictive control,''
  \emph{PhD thesis, Department of ElectricalEngineering, Link\"oping
  University, Sweden}, 2003.

\bibitem{boyd_convex_2004}
S.~P. Boyd and L.~Vandenberghe, \emph{\BIBforeignlanguage{en}{Convex
  optimization}}.\hskip 1em plus 0.5em minus 0.4em\relax Cambridge, UK ; New
  York: Cambridge University Press, 2004.

\bibitem{amo_alonso_distributedLQR_2020}
C.~{Amo Alonso}, D.~Ho, and J.~Maestre, ``Distributed linear quadratic
  regulator robust to communication dropouts,'' \emph{IFAC-PapersOnLine},
  vol.~53, no.~2, pp. 3072 -- 3078, 2020.

\bibitem{wang_large-scale_2021}
\BIBentryALTinterwordspacing
H.~{Wang} and J.~{Anderson}, ``Large-scale system identification using a
  randomized {SVD},'' 2021. [Online]. Available:
  \url{https://arxiv.org/abs/2109.02703}
\BIBentrySTDinterwordspacing

\bibitem{yu_localized_2020}
\BIBentryALTinterwordspacing
J.~Yu, Y.-S. Wang, and J.~Anderson, ``Localized and distributed $\mathcal{H}_2$
  state feedback control,'' 2020. [Online]. Available:
  \url{https://arxiv.org/abs/2010.02440}
\BIBentrySTDinterwordspacing

\bibitem{gilbert_linear_1991}
E.~Gilbert and K.~Tan, ``Linear systems with state and control constraints: the
  theory and applications of maximal output admissible sets,'' \emph{IEEE
  Trans. Autom. Control}, vol.~36, no.~9, pp. 1008 --1020, 1991.

\bibitem{costantini_decomposition_2018}
G.~Costantini, R.~Rostami, and D.~Gorges, ``Decomposition {Methods} for
  {Distributed} {Quadratic} {Programming} with {Application} to {Distributed}
  {Model} {Predictive} {Control},'' in \emph{IEEE Proc. {Annu.} {Allerton}
  {Conf.} {Commun.}, {Control}, {Comput.}}, 2018, pp. 943 -- 950.

\bibitem{slstoolbox}
\BIBentryALTinterwordspacing
J.~S. Li, ``{SLS-MATLAB}: Matlab toolbox for system level synthesis,'' 2019.
  [Online]. Available: \url{https://github.com/sls-caltech/sls-code}
\BIBentrySTDinterwordspacing

\bibitem{cvx}
M.~Grant and S.~Boyd, ``{CVX}: Matlab software for disciplined convex
  programming, version 2.1,'' \url{http://cvxr.com/cvx}, Mar. 2014.

\end{thebibliography}

\end{document}